\documentclass[letterpaper,11pt]{article}
\usepackage{amsmath,amssymb,amsfonts,amsthm}
\usepackage{setspace}
\usepackage{bbm}
\usepackage{color}
\usepackage{enumerate}
\usepackage{mathrsfs}
\usepackage[longnamesfirst]{natbib}
\usepackage[bottom]{footmisc}
\usepackage{graphicx}
\usepackage{subcaption}
\usepackage{hyperref}
\usepackage{cleveref}
\usepackage[shortlabels]{enumitem}
\usepackage[dvipsnames,svgnames,table]{xcolor}
\usepackage{multirow}
\usepackage[margin=1.2in]{geometry}
\usepackage{enumitem}
\usepackage{appendix}
\usepackage{authblk}
\usepackage{booktabs}
\usepackage{parskip}
\newtheoremstyle{paperplain}
  {0.2\baselineskip}  
  {0.2\baselineskip} 
  {\itshape}        
  {}                
  {\bfseries}       
  {.}               
  {0.5em}           
  {}

\theoremstyle{paperplain}
\newtheorem{theorem}{Theorem}
\newtheorem{lemma}{Lemma}
\newtheorem{proposition}{Proposition}

\newtheorem{assumption}{Assumption}

\newtheorem{remark}{Remark}[section]

\crefname{assumption}{Assumption}{Assumptions}
\crefname{subassumptioni}{Assumption}{Assumptions}
\crefname{section}{Section}{Sections}
\crefname{theorem}{Theorem}{Theorems}
\crefname{lemma}{Lemma}{Lemmas}

\newlist{subassumption}{enumerate}{1}
\setlist[subassumption,1]{
  label=(\roman*),   
  ref=\theassumption(\roman*), 
  leftmargin=2em
}

\newcommand{\mar}{\text{MAR}}
\newcommand{\an}{\text{AN}}
\newcommand{\E}{\operatorname{\mathbb{E}}}

\newcommand{\Pb}{\operatorname{\mathbb{P}}}
\newcommand{\KL}{\operatorname{KL}}

\newcommand{\R}{\mathbb{R}}

\newcommand{\vech}{\operatorname{vech}}
\newcommand{\I}{\operatorname{\mathbb{I}}}

\DeclareMathOperator*{\argmin}{arg\,min}

\hypersetup{
     colorlinks   = true,
     citecolor    = blue!50!black,
     urlcolor     = blue!50!black,
     linkcolor    = blue!50!black
}

\title{\textbf{
     Testing selection on observables
     in parametric models
     with refreshment samples}\footnote{The project described in this publication relies on data from surveys administered by the Understanding America Study (UAS), which is maintained by the Center for Economic and Social Research (CESR) at the University of Southern California. The project was supported by the National Institute on Aging of the National Institutes of Health and, in part, by the Social Security Administration under Award Number U01AG077280. The content is solely the responsibility of the authors and does not necessarily represent the official views of the National Institutes of Health, USC, or UAS. The data used in the empirical application are publicly available from the UAS website subject to signing a data use agreement.
}
}

\date{\today}

\author[1]{\textsc{Grigory Franguridi}}
\author[1]{\textsc{Arie Kapteyn}}

\affil[1]{Center for Economic and Social Research, University of Southern California

Franguridi: franguri@usc.edu; Kapteyn: kapteyn@usc.edu}

\begin{document}

\maketitle

\begin{abstract}
\linespread{1.2}
    In panels with sample selection (that may occur due to attrition, nonresponse, etc.), the assumption of selection on observables (missing at random, MAR) is commonly imposed despite often being implausible.
    However, this assumption becomes testable when a refreshment sample is available.
    We develop a statistical test of MAR based on a distance between two estimated distributions: one obtained using the standard inverse probability weighting (IPW) that is valid under MAR and the other obtained using an alternative weighting that is valid under a weaker assumption of additive nonignorability of \citet{hirano2001combining}.
    This test implicitly compares the distribution of the IPW-weighted sample in the attrition period with the distribution of the refreshment sample, which coincide if the MAR assumption holds.
    We establish that, when the input distributions are parametric, our test statistic converges to the generalized chi-squared distribution under the null of MAR.
    This limit distribution can be estimated using the recursive formulas derived by \citet{franguridi2025raking}.
    We illustrate the performance of our test in Monte Carlo simulations.
    Finally, we apply our test to an empirical example using a subsample of the Understanding America Study (UAS) dataset.

\medskip


\medskip

\end{abstract}

\newpage

\section{Introduction}

In the last 60 years or so, both empirical applications of panel data models and associated methodological developments have seen tremendous growth \citep[e.g.,][]{sarafidis2021celebrating}. One of the most intricate problems in using panel data is attrition, i.e., units dropping out of the sample in a nonrandom way.
This reduces the effective sample size and can bias the resulting estimates.
These biases can be completely removed when units are \textit{missing at random} (MAR; also called \textit{selection on observables}), i.e., when the probability of attrition depends only on observed variables. In that case, inverse probability weighting of the sample of stayers leads to unbiased estimates.
However, the MAR assumption can be restrictive, and panel data alone is not sufficient to determine its plausibility.

To compensate for attrition, researchers often collect so-called \emph{refreshment samples}, i.e., samples that are designed to ``replace'' the missing units \citep{kish1959variances}.
Refreshment samples are used in many empirical settings, including survey panels \citep{deng2013handling,si2015semi,franguridi2024closed,franguridi2025inference}, public transportation data \citep{ridder1992empirical,zheng2025semiparametric}, and retail scanner data \citep{chen2017retail}.
With a refreshment sample, the MAR assumption underlying the inverse probability weighting becomes testable.

We develop a statistical test of MAR that relies on a modeling framework with \textit{additively nonignorable} (AN) attrition introduced by \citet{hirano2001combining} and further used by \citet{nevo2003using,bhattacharya2008inference,hoonhout2019nonignorable,franguridi2025raking,franguridi2025robust} and many others.
The AN assumption allows the attrition to depend on yet unobserved variables and includes MAR as a special case.
We propose parameterizing the three directly estimable distributions (the first-period distribution, the balanced panel distribution, and the refreshment sample distribution) and compute the Hellinger distance between the joint distribution of data estimated under the MAR assumption (the inverse probability weighted distribution) and that estimated under the AN assumption as a test statistic.
To compute the AN distribution, we use the raking algorithm in \citet{franguridi2025raking}.
We show that our statistic has an asymptotically generalized chi-squared distribution, which can be estimated using the recursive formulas also derived in \citet{franguridi2025raking}.

The remainder of the paper is organized as follows.
\cref{sec:framework} introduces the framework.
\cref{sec:testing} describes the test statistic and derives its asymptotic distribution.
\cref{sec:mc-sim} illustrates the size and power of our test in a Monte Carlo simulation.
\cref{sec:empirical} uses the test to assess the MAR assumption in a subset of the Understanding America Study (UAS) panel survey.
\cref{sec:conclusion} concludes.
Finally, the Appendix specializes our results to the cases of Gaussian data and discrete data and describes the construction of the cognition variable.

\section{Framework}\label{sec:framework}

Consider a two-period panel in which units may drop out of the sample in period 2.
Let $Z_{it}=(Y_{it},X_{it})$ be stacked outcomes and covariates for unit $i$ in period $t=1,2$.
In period 1, data $Z_{i1}$ for all the units $i=1,\dots,n_1$ are observed, while in period 2, data $Z_{i2}$ are observed only for units who stay in the sample, which we denote by $S_i=1$.
Using only data of the selected sample $\{i:\,S_i=1\}$ may lead to significant bias. To reduce this bias, researchers often assume that the selection probability
\begin{align*}
p(z_1,z_2)=\Pb(S_i=1|Z_{1i}=z_1,Z_{2i}=z_2)
\end{align*}
only depends on the observables $z_1$, not the potentially missing $z_2$.
This allows estimating the selection probability by running a (possibly nonparametric) regression of $S_i$ on $Z_{1i}$ and then using the estimate as the inverse weight for the selected sample.

The ``selection on observables'' (MAR) assumption may be implausible in many empirical settings.
Consider, for example, a panel survey.
It is likely that survey respondents decide whether to continue participating in the survey depending on data yet unobserved to the researcher.
For example, if $Z_{it}$ is respondent $i$'s income at time $t$, then the respondent $i$ may drop out of the sample after experiencing a negative income shock, i.e., when $Z_{i2}$ is significantly smaller than $Z_{i1}$ (before reporting $Z_{i2}$ to the researcher).

Without auxiliary information, the MAR assumption is untestable.
However, in survey panels, researchers often collect so-called \emph{refreshment samples}, i.e., samples that are designed to ``replace'' the missing units.
We denote a refreshment sample in period 2 by $Z_{i2}^r$, $i=1,\dots,n_r$.
What makes it a ``refreshment'' sample is being an independent sample from the target period-2 marginal distribution, i.e., $Z_{i2}^r \overset{d}{=} Z_{j2}$.
Our key observation is that, with a refreshment sample, the MAR assumption becomes testable.
This is because the selection probability $p(z_1,z_2)$ is identifiable with a refreshment sample under the assumption of \emph{additive nonignorability} (AN) introduced by \citet{hirano2001combining}:
\begin{align}
     p(z_1,z_2) = \exp(k_1(z_1)+k_2(z_2)) \text{ for some unknown functions } k_1,k_2. \label{eq:AN}
\end{align}
This assumption is significantly weaker than MAR and can be used to develop a computationally feasible test of the null hypothesis of MAR against the alternative hypothesis of AN.\footnote{We focus on the exponential link function $G(x)=\exp(x)$ because it leads to a computationally convenient procedure based on raking, as shown by \citet{franguridi2025raking}. In principle, all the theoretical results of this paper can be generalized to any smooth link function $G$.}
Under AN, \citet{hirano2001combining} showed that the joint density $f_\an(z_1,z_2)$ of $(Z_1,Z_2)$ is identified and can be recovered from the following three identified objects: the selected density $f^s(z_1,z_2) = f(z_1,z_2|S=1)$, the density $f_1(z_1)$ of $Z_1$ (identified with the first-period data where there is no attrition), and the density $f_2(z_2)$ of $Z_2$ (identified with the refreshment sample).
Specifically, $f_\an$ is the solution to the functional projection problem
\begin{align}
	\min_{\tilde f\in\Pi(f_1,f_2)} \KL(\tilde f,f^s), \label{eq:KL-projection}
\end{align}
where $\Pi(f_1,f_2)$ is the set of joint densities with (multivariate) marginals $f_1$ and $f_2$ and
\begin{align*}
\KL(f,f^s) = \E_{f} \left[ \log \frac{f(Z_1,Z_2)}{f^s(Z_1,Z_2)} \right]
\end{align*}
is the Kullback-Leibler (KL) divergence between $f$ and $f^s$.\footnote{This result also holds when the data contain a mix of discrete and continuous variables, in which case the densities involved should be understood as densities with respect to an appropriate product of Lebesgue and counting measures.}
In other words, $f_\an$ is the closest (in the KL sense) density to $f^s$ among densities with marginals $f_1$ and $f_2$.
Notice how the functions $k_1,k_2$ do not appear explicitly in this characterization of $f_\an$.
They do appear in the dual formulation of this problem, but this formulation is not required for this paper.

\section{Testing procedure}\label{sec:testing}

\subsection{Test statistic}


It is natural to seek evidence against MAR in a statistical distance between distributions $f_\an$ and $f_\mar$.
We consider the squared Hellinger distance due to its analytical tractability in our framework,\footnote{In principle, we could consider a test statistic based on any other statistical distance that depends smoothly on the densities of the input distributions, such as the Kullback-Leibler divergence, but the resulting formulas for the asymptotic distribution of such a statistic would be more complicated.}
\begin{align*}
    \hat T_n=\frac{1}{2} \int \left( \sqrt{\hat f_\an(z_1,z_2)}-\sqrt{\hat f_\mar(z_1,z_2)}\right)^2 \, d z_1 dz_2,
\end{align*}
where $\hat f_\an$ and $\hat f_\mar$ are estimates of $f_\an$ and $f_\mar$, respectively.

Estimation of $f_\mar$ is easy. Indeed, the Bayes formula and the MAR assumption imply
\begin{align*}
    f_\mar(z_1,z_2) = \frac{\Pb(S=1) f^s(z_1,z_2)}{\Pb(S=1|Z_1=z_1,Z_2=z_2)}  = \frac{\Pb(S=1) f^s(z_1,z_2)}{\Pb(S=1|Z_1=z_1)} .
\end{align*}
Integrating with respect to $z_2$ and taking into account that the first marginal of $f_\mar$ is $f_1$, we obtain
\begin{align*}
    f_\mar(z_1,z_2) = \frac{f_1(z_1) f^s(z_1,z_2)}{\int f^s(z_1,z_2)\, dz_2}.
\end{align*}
Hence, we can use the plug-in estimator
\[
    \hat f_\mar(z_1,z_2) = \frac{\hat f_1(z_1) \hat f^s(z_1,z_2)}{\int \hat f^s(z_1,z_2)\, dz_2},
\]
where $\hat f_1, \hat f_2, \hat f^s$ are chosen estimators of $f_1,f_2,f^s$, respectively.
Notice that this formula forces $\hat f_\mar$ to have the first marginal $\hat f_1$.
Incidentally, $\hat f_\mar$ can be obtained from $\hat f^s$ after just one iteration of the raking algorithm described in \Cref{sec:raking}.

Estimating $f_\an$ is a more difficult problem, and we defer its exposition to \Cref{sec:raking}.

\begin{remark}
    While the MAR assumption corresponds to the lack of dependence of the selection probability on $Z_2$, we may want to test its symmetric analog $\Pb(S=1|Z_1, Z_2)=\Pb(S=1|Z_2)$ instead. This can be handled easily because the target distribution in this case can be written as
    \begin{align*}
        f(z_1,z_2)= \frac{f_2(z_2) f^s(z_1,z_2)}{\int f^s(z_1,z_2) \, dz_1},
    \end{align*}
    where $f_2$ is identified using the refreshment sample. Besides, identification under this assumption is possible even without refreshment samples under additional parametric restrictions; see \citet{hausman1979attrition}.
\end{remark}

\subsection{Asymptotic distribution}

We now assume the densities $f_1,f_2,f^s$ are parametric and derive the asymptotic distribution of our test statistic under the null hypothesis of MAR.

Let $\lambda$ be a common dominating measure for $Z_1, Z_2$. For example, if both $Z_1$ and $Z_2$ contain one binary outcome and one continuous covariate, then $\lambda=\lambda_1\otimes\lambda_2$, where $\lambda_k$ is the product of the discrete uniform distribution on $\{0,1\}$ and the Lebesgue measure on $\R$.

We impose that the densities $f_1,f_2,$ and $f^s$ belong to parametric families $f_1(\cdot,\gamma_1)$, $f_2(\cdot,\gamma_2)$, and $f^s(\cdot,\gamma^s)$, where 
\[
\gamma = (\gamma_1',\gamma_2',(\gamma^s)')'\in\Gamma\subset\mathbb R^p
\]
is a finite-dimensional parameter with the true value $\gamma_0$.
For $j\in\{\mathrm{AN},\mathrm{MAR}\}$, let $f_j(\cdot,\gamma)$ denote the density induced by the parameter $\gamma$, i.e.,
\begin{align*}
    f_\mar(z_1,z_2,\gamma) &= \frac{f_1(z_1,z_2,\gamma_1) f^s(z_1,z_2,\gamma^s)}{\int f^s(z_1,z_2,\gamma^s) \lambda(dz_2)}, \\
    f_\an(\cdot,\gamma) &= \argmin_{\tilde f \in \Pi(f_1(\cdot,\gamma_1),f_2(\cdot,\gamma_2))} \KL(\tilde f, f^s(\cdot,\gamma^s)).
\end{align*}
Finally, define
\[
T(\gamma)
=
\frac12
\int
\left\{
\sqrt{f_{\mathrm{AN}}(z_1,z_2,\gamma)}
-
\sqrt{f_{\mathrm{MAR}}(z_1,z_2,\gamma)}
\right\}^2
\,d\lambda(z_1,z_2).
\]
We impose the following assumption.

\begin{assumption}\label{ass:hellinger-limit}
\begin{enumerate}[(i)]
\item
The parameter space $\Gamma$ contains an open neighborhood of
$\gamma_0$. For every $\gamma$ in this neighborhood, $f_{\mathrm{MAR}}(\cdot,\gamma)$ and $f_{\mathrm{AN}}(\cdot,\gamma)$ are well-defined densities with respect to the dominating measure $\lambda$.

\item
Under the null hypothesis of MAR,
\[
f_{\mathrm{AN}}(z,\gamma_0)
=
f_{\mathrm{MAR}}(z,\gamma_0)
=:f_0(z)
\qquad
\text{for $\lambda$-almost every }z.
\]

\item
For each $j\in\{\mathrm{AN},\mathrm{MAR}\}$, the map
$\gamma\mapsto f_j(\cdot,\gamma)$ is differentiable in quadratic
mean at $\gamma_0$, i.e., there exists a measurable vector-valued
function $s_j(\cdot,\gamma_0)\in L^2(f_0)^p$ such that, as $h\to0$,
\[
\int
\left[
\sqrt{f_j(z_1,z_2,\gamma_0+h)}
-
\sqrt{f_0(z_1,z_2)}
-
\frac12
\sqrt{f_0(z_1,z_2)}
s_j(z_1,z_2,\gamma_0)'h
\right]^2
\,d\lambda(z_1,z_2)
=
o(\|h\|^2).
\]
Notice that
\[
s_j(z_1,z_2,\gamma_0)
=
\nabla_\gamma\log f_j(z_1,z_2,\gamma_0)
\qquad
f_0\text{-almost surely}
\]
whenever the derivative on the right exists.

\item
The estimator $\widehat\gamma$ satisfies
\[
\sqrt n(\widehat\gamma-\gamma_0)
\rightsquigarrow
g_0,
\qquad
g_0\sim N(0,\Omega_0)
\]
for some positive semidefinite matrix $\Omega_0$.
\end{enumerate}
\end{assumption}

\Cref{ass:hellinger-limit}(i) imposes that $\gamma_0$ belongs to the interior of its parameter space $\Gamma$ and that slight deviations from $\gamma_0$ still allow well-defined densities $f_\mar$ and $f_\an$.
\Cref{ass:hellinger-limit}(ii) states that the null hypothesis of MAR means that the AN and MAR densities coincide.
\Cref{ass:hellinger-limit}(iii) is a standard condition on a smoothness of densities $f_\mar$ and $f_\an$ with respect to $\gamma$ around the truth. A simple sufficient condition is continuous differentiability of the square root of these densities along with suitable domination; see, e.g., Lemma 7.6 in \citet{van2000asymptotic}.
Finally, \Cref{ass:hellinger-limit}(iv) imposes a Gaussian limit on the estimator $\hat\gamma$, which holds, for example, if $\hat\gamma$ is the maximum likelihood estimator. Here $n=n_1+n_r$ is the combined sample size of the first period and the refreshment sample. \Cref{ass:hellinger-limit}(iv) implicitly requires that $\min(n_1,n_r) \to \infty$.

We are now ready to state our main result.

\begin{theorem}\label{thm:hellinger-limit}
Suppose \cref{ass:hellinger-limit} holds and define
\[
\widehat T_n=T(\widehat\gamma).
\]
Then, under the null hypothesis of MAR,
\[
n\widehat T_n
\rightsquigarrow
\frac12 g_0'H_0g_0,
\qquad
g_0\sim N(0,\Omega_0),
\]
where
\[
H_0
=
\frac14
\E_{f_0}
\left[
\left\{
s_{\mathrm{AN}}(Z_1,Z_2,\gamma_0)
-
s_{\mathrm{MAR}}(Z_1,Z_2,\gamma_0)
\right\}
\left\{
s_{\mathrm{AN}}(Z_1,Z_2,\gamma_0)
-
s_{\mathrm{MAR}}(Z_1,Z_2,\gamma_0)
\right\}'
\right].
\]
\end{theorem}

\begin{proof}
Let $h_n=\widehat\gamma-\gamma_0$
and $\Delta s(z)= s_{\mathrm{AN}}(z,\gamma_0) - s_{\mathrm{MAR}}(z,\gamma_0)$.
By \cref{ass:hellinger-limit}(ii) and \cref{ass:hellinger-limit}(iii),
\[
\sqrt{f_{\mathrm{AN}}(\cdot,\gamma_0+h_n)}
-
\sqrt{f_{\mathrm{MAR}}(\cdot,\gamma_0+h_n)}
=
\frac12\sqrt{f_0}\,\Delta s'h_n+r_n,
\]
where
\[
\|r_n\|_{L^2(\lambda)}
=
o_p(\|h_n\|).
\]
Since $\Delta s\in L^2(f_0)^p$,
\[
\left\|
\sqrt{f_0}\,\Delta s'h_n
\right\|_{L^2(\lambda)}
=
O_p(\|h_n\|).
\]
We have
\begin{align*}
\widehat T_n
&=
\frac12
\left\|
\frac12\sqrt{f_0}\,\Delta s'h_n+r_n
\right\|_{L^2(\lambda)}^2 \\
&=
\frac18
h_n'
\E_{f_0}\!\left[\Delta s(Z)\Delta s(Z)'\right]
h_n
+
o_p(\|h_n\|^2) \\
&=
\frac12h_n'H_0h_n+o_p(n^{-1}),
\end{align*}
where the second equation uses the Cauchy--Schwarz inequality.
Therefore,
\[
n\widehat T_n
=
\frac12
\left\{\sqrt n(\widehat\gamma-\gamma_0)\right\}'
H_0
\left\{\sqrt n(\widehat\gamma-\gamma_0)\right\}
+
o_p(1).
\]
Using \cref{ass:hellinger-limit}(iv) and the continuous mapping theorem completes the proof.
\end{proof}

One convenient feature of the asymptotic distribution of our test statistic is that it does not depend on the second derivatives of $f_\an$ and $f_\mar$ with respect to $\gamma$, which would be the case for other distances such as the KL divergence.

\subsection{Estimator of the AN density}\label{sec:raking}

Here we describe a raking-based estimator of $f_\an$ developed in \citet{franguridi2025raking}.
To this end, we introduce some notation.
First, let $\I_1$ and $\I_2$ be integration operators defined for any integrable function $f=f(z_1,z_2)$ by
\begin{align*}
    [\I_1 f](z_2) = \int f(z_1,z_2)\, \lambda(dz_1), \qquad [\I_2 f](z_1) = \int f(z_1,z_2)\, \lambda(dz_2).
\end{align*}
Second, let $\Pi_1$ and $\Pi_2$ be the KL projection operators on the sets of distributions with the first marginal $f_1$ and the second marginal $f_2$, respectively,\footnote{See, e.g., \citet{csiszar1975divergence} for a proof that these operators indeed define the KL projection onto the set of distributions with one fixed marginal.} i.e., for any density $f$,
\begin{align*}
    [\Pi_1 f](z_1,z_2) = \frac{f_1(z_1) f(z_1,z_2)}{[\I_2 f](z_1)}, \\
    [\Pi_2 f](z_1,z_2) = \frac{f_2(z_2) f(z_1,z_2)}{[\I_1 f](z_2)}.
\end{align*}
Third, define the composition operator $\Pi = \Pi_2 \circ \Pi_1$, i.e.,
\begin{align}
        [\Pi f](z_1,z_2)= m & \left[ f_2(z_2,\gamma_2), \right. \notag \\
        &m \left(f_1(z_1,\gamma_1), f(z_1,z_2), \I_2 f(z_1,\cdot)\right), \notag \\
        &\left. \I_1 m\left( f_1(\cdot,\gamma_1), f(\cdot,z_2), \I_2 f(\cdot,\cdot) \right) \right],
\end{align}
where $m(a,b,c)=ab/c$.
The raking procedure (also called iterative proportional fitting or Sinkhorn's algorithm) is an iteration of the operator $\Pi$ starting from a given density function.
It turns out that under weak conditions, the raking procedure initialized at $f^s$ converges to the solution of the KL projection problem \eqref{eq:KL-projection}, i.e.,
\begin{align*}
    \lim_{T \to \infty} \Pi^{(T)} f^s = f_\an,
\end{align*}
where $\Pi^{(T)}$ is the T-fold iteration of $\Pi$ and the convergence is understood in the $L^1$ sense, see \citet{ruschendorf1995convergence,franguridi2025raking}.
In other words, raking starts with $f^s$ and alternates between projection onto the set of distributions with the first marginal $f_1$ and projection onto the set of distributions with the second marginal $f_2$.

Dropping the variables $z_1,z_2$ and reflecting the dependence of the previous iteration on $\gamma$, we can write the $(t+1)$-th raking iteration as
\begin{align}
    f^{(t+1)}(\gamma) = m & \left[ f_2(\gamma_2), \right. \notag \\
    &m \left(f_1(\gamma_1), f^{(t)}(\gamma), \I_2 f^{(t)}(\gamma) \right), \notag \\
    &\left. \I_1  m\left( f_1(\gamma_1), f^{(t)}(\gamma), \I_2 f^{(t)}(\gamma) \right) \right], \label{eq:raking-recursion}
\end{align}
with the initial condition
\begin{align*}
    f^{(0)}(\gamma) = f^s(\gamma).
\end{align*}

The raking estimator of \citet{franguridi2025raking} is then a sample analog of this procedure,
\begin{align*}
    \hat f_\an = \hat\Pi^{(T)} \hat f^s,
\end{align*}
where $\hat\Pi$ is the sample raking operator that uses estimators $\hat f_1,\hat f_2$ in place of $f_1,f_2$; $\hat f^s$ is an estimator of $f^s$; and $T$ is the number of iterations chosen in advance.\footnote{The speed of convergence of the raking iterations was studied by \citet{ruschendorf1995convergence}.}

\subsection{Critical values}

The key object needed to find the critical values of our test is the Hessian $H(\gamma_0)$.
For its plug-in estimation, we need the estimators of the derivatives of $f_\an$ and $f_\mar$ with respect to $\gamma$.
Fortunately, these estimators were developed in \citet{franguridi2025raking}.
In this section, we will formulate their methodology using the operator notation that is more concise and suitable for programming.

Denote the first derivatives of $m(a,b,c)$ with respect to its arguments by
\begin{align*}
    m_1(a,b,c) = b/c, \qquad m_2(a,b,c) = a/c, \qquad m_3(a,b,c) = -ab/c^2.
\end{align*}
Differentiating the raking recursion formula \eqref{eq:raking-recursion} with respect to $\gamma$ yields
\begin{align}
\nabla_\gamma f^{(t+1)}
=
m_1(f_2,h_t,\I_1 h_t)\,\nabla_\gamma f_2
+
m_2(f_2,h_t,\I_1 h_t)\,\nabla_\gamma h_t
+
m_3(f_2,h_t,\I_1 h_t)\I_1 \nabla_\gamma h_t, \label{eq:Df-recursion}
\end{align}
where
\begin{align*}
    h_t &=m(f_1,f^{(t)}, \I_2 f^{(t)}), \\
\nabla_\gamma h_t &= m_1(f_1,f^{(t)},\I_2 f^{(t)})\,\nabla_\gamma f_1
+
m_2(f_1,f^{(t)},\I_2 f^{(t)})\,\nabla_\gamma f^{(t)}
+
m_3(f_1,f^{(t)},\I_2 f^{(t)})\I_2 \nabla_\gamma f^{(t)}.
\end{align*}
The right-hand side of \eqref{eq:Df-recursion} depends on the fixed functions $f_1,f_2,\nabla_\gamma f_1,\nabla_\gamma f_2$ and the objects from the previous iteration $f^{(t)}$ and $\nabla_\gamma f^{(t)}$.
Therefore, we can write
\begin{align*}
    (f^{(t+1)}, \, \nabla_\gamma f^{(t+1)}) = \mathcal{R}_{\gamma_1,\gamma_2} (f^{(t)}, \, \nabla_\gamma f^{(t)}),
\end{align*}
where $\mathcal{R}_{\gamma_1,\gamma_2}$ is the operator mapping a pair of functions of $z_1,z_2$ to another such pair of functions and which only depends on the fixed functions $f_1,f_2,\nabla_\gamma f_1,\nabla_\gamma f_2$.
The subscript $\gamma_1,\gamma_2$ emphasizes that the operator is applied at specific values $\gamma_1,\gamma_2$ where the functions and their derivatives are being calculated.
Taking into account the initial condition $f^{(0)}=f^s$, we have
\begin{align*}
    (f^{(T)}, \, \nabla_\gamma f^{(T)}) = \mathcal{R}_{\gamma_1,\gamma_2}^{(T)} (f^s, \, \nabla_\gamma f^s).
\end{align*}
In particular, since $\hat f = f^{(T)}(\hat\gamma)$,
\begin{align*}
    (\hat f, \, \nabla_\gamma {\hat f}) = \mathcal{R}_{\hat \gamma_1,\hat \gamma_2}^{(T)} (f^s(\hat\gamma^s), \, \nabla_\gamma f^s(\hat\gamma^s)).
\end{align*}

Finally, since $f_\mar = m(f_1,f^s, \I_2 f^s)$, we have
\begin{align*}
    \nabla_\gamma f_\mar = m_1(f_1,f^s,\I_2 f^s)\,\nabla_\gamma f_1
    +
    m_2(f_1,f^s,\I_2 f^s)\,\nabla_\gamma f^s
    +
    m_3(f_1,f^s,\I_2 f^s)\I_2 \nabla_\gamma f^s.
\end{align*}

\section{Monte Carlo simulation}\label{sec:mc-sim}

In this section, we evaluate the finite-sample performance of our test in a set of Monte Carlo simulations for a Gaussian data-generating process.

We model the latent (target) distribution $f$ of $Z=(Z_1',Z_2')'$ as
\[
Z \sim N(\mu,\Sigma),
\]
where the mean \(\mu=0\) and the covariance matrix
\[
\Sigma=
\begin{pmatrix}
I_d & \rho I_d\\
\rho I_d & I_d
\end{pmatrix},
\qquad \rho=0.5.
\]
We then model the selection probability as
\[
\Pb(S=1\mid Z_1,Z_2)=\exp\!\left(\alpha-Z_1' B_1 Z_1 - Z_2' B_2' Z_2\right),
\]
with \(B_1=0\).
These choices imply that the selected density $f^s$ is also Gaussian.
The refreshment sample is drawn from the second marginal of $f$ and always has the same number of observations $n$ as the first-period sample.

We consider two DGPs, one for evaluating the size and the other for evaluating the power of our test.
The first DGP $H_0$ sets \(B_2=0\), leading to a constant selection probability, the simplest case satisfying the null hypothesis of MAR.
The second DGP $H_1$ sets
\[
B_2=0.05\,I_d,
\]
so that
\[
\Pb(S=1\mid Z_1,Z_2)= \exp\!\left( \alpha-0.05\,\|Z_2\|^2 \right).
\]
In this case, selection depends on the unobserved component \(Z_2\), generating a mechanism that violates MAR but satisfies AN.
Under both $H_0$ and $H_1$, the intercept \(\alpha\) is chosen to attain the target attrition rate \(\Pb(S=0)=0.2\).

To obtain the critical values for our test, we employ the following procedure that involves the nonparametric bootstrap for the parameter estimator
\[
\hat\gamma = (\hat\mu_1',\vech(\hat\Sigma_1),\hat\mu_2',\vech(\hat\Sigma_2),(\hat\mu^{\,s})',\vech(\hat\Sigma^{\,s}))'.
\]
First, we compute the estimate $\hat H(\hat\gamma)$ of the Hessian of the test statistic using the formulas in \Cref{sec:gaussian}.
Second, we generate bootstrap draws \(\hat \gamma^{*,b}\), $b=1,\dots, B$ and compute the bootstrap values of the test statistic
\begin{align*}
    \hat T^{*,b} = \frac{1}{2} (\hat\gamma^{*,b}-\hat\gamma)' \hat H(\hat\gamma) (\hat\gamma^{*,b}-\hat\gamma), \quad b=1,\dots, B.
\end{align*}
Finally, we reject the null hypothesis of MAR if the test statistic $\hat T$ exceeds the $1-\alpha$ sample quantile of $\hat T^{*,1},\dots,\hat T^{*,B}$. We use $B=300$ bootstrap samples.

\Cref{tab:mc} reports the empirical rejection rates of our test based on $1000$ Monte Carlo replications across a range of values of $n$ and $d$.
Under the null hypothesis $H_0$, the rejection frequencies are close to the nominal 10\% level across all sample sizes and dimensions, indicating that the test provides excellent size control.
Under the alternative hypothesis $H_1$, the test exhibits substantial power even at a relatively small sample size $n=2000$.
Power then increases rapidly with sample size, reaching $100\%$ when $n=10000$.
Overall, the results suggest that the test maintains accurate size while achieving high power against nonignorable attrition alternatives.

\begin{table}[ht]
    \centering
    \begin{tabular}{c|cc|cc}
    & \multicolumn{2}{c|}{$H_0$} & \multicolumn{2}{c}{$H_1$} \\
    $n$ & $d=2$ & $d=4$ & $d=2$ & $d=4$ \\
    \midrule
    2000  & 0.114 & 0.093 & 0.583 & 0.676 \\
    5000  & 0.101 & 0.116 & 0.937 & 0.994 \\
    10000 & 0.101 & 0.090 & 1.000 & 1.000 \\
    \bottomrule
    \end{tabular}
    \caption{Empirical rejection rates (size and power) based on $1000$ Monte Carlo replications.}
    \label{tab:mc}
\end{table}

\section{Empirical illustration}\label{sec:empirical}

We apply the framework developed here to data from the Understanding America Study (UAS), a panel of approximately 15,000 U.S. respondents. The UAS started in 2014 and is still growing. This means that new batches of respondents are regularly added, both to expand the panel and to replace respondents lost to attrition. The panel collects a large amount of information in a two-year cycle, including health, health behavior and health insurance, financial status and labor market behavior, cognition, and personality. The motivation for collecting a large amount of information across different domains is that it allows much richer analysis than when a dataset concentrates on only one dimension, e.g., economics or cognition.

The example we consider is how wealth holdings are affected by income, education, and cognition. The latter variable is an aggregate obtained from factor analysis of the outcomes of five different cognitive tests. Appendix \ref{sec:app-cognition} describes the construction of the cognition variable. As noted, core information is collected in a two-year cycle. UAS respondents can join at different times, which means that respondents may answer core surveys at different points in a two-year interval.
We consider answers of respondents in the wave that spans 4 July 2021 to 4 July 2023, which, for simplicity, we refer to as Wave 1, and the wave that spans 5 July 2023 to 6 July 2025, referred to as Wave 2.
The wave that is currently in the field will only be complete by July 2027 and hence is not suitable for the analysis in this paper.

To implement the test in this empirical example, the four explanatory variables are all discretized into three categories. Education is categorized as High School or less, Some College, or College or more. Income is categorized into three terciles, and so are cognition and financial wealth.
Wave 1 has 9,725 observations, of whom 7,432 also appear in Wave 2. In addition, Wave 2 has 2,264 fresh respondents (the refreshment sample).

Our test strongly rejects the null of no selection on unobserved Wave 2 variables at any reasonable level, with a p-value of $5\cdot 10^{-6}$. This result suggests that Wave 2 attrition is related to the variables in Wave 2. The test cannot tell us which of these unobserved Wave 2 variables has affected the probability of dropping out of the sample.
Responding to surveys is a cognitively demanding activity that takes increasing effort when cognition declines. Thus, one possible explanation is that respondents with declining cognition are more likely to have dropped out of the sample.

\section{Conclusion}\label{sec:conclusion}

The commonly imposed assumption of selection on observables (MAR) may be implausible in many empirical settings and is generally untestable.
However, this assumption becomes testable when a refreshment sample is available, which is the case for many panel surveys used in empirical research.
We propose to test the MAR assumption against a weaker alternative of additive nonignorability introduced in \citet{hirano2001combining}.
Our test statistic implicitly measures the distributional deviation of the refreshment sample from the inverse probability weighted selected sample in the second period.
When the input distributions are parameterized, the computation of the test statistic and the critical values can be performed using the raking-based recursive algorithm developed in \citet{franguridi2025raking}.

\bibliographystyle{ecta}
\bibliography{references}

\onehalfspacing
\frenchspacing
\newpage  
\part*{Appendix}
\appendix

\section{The Gaussian case}\label{sec:gaussian}

In this section, we specialize our testing procedure to the case where the distributions $f_1,f_2,f^s$ are assumed to be Gaussian.

\subsection{The test statistic}

The squared Hellinger distance between distributions $N(\mu_\an,\Sigma_\an)$ and $N(\mu_\mar,\Sigma_\mar)$ has a closed-form expression
\begin{align*}
    T =1-\frac{|\Sigma_\an|^{1/4}|\Sigma_\mar|^{1/4}
}{
\left|\frac{\Sigma_\an+\Sigma_\mar}{2}\right|^{1/2}
}
\exp\!\left(
-\frac18(\mu_\an-\mu_\mar)'
\left(\frac{\Sigma_\an+\Sigma_\mar}{2}\right)^{-1}
(\mu_\an-\mu_\mar)
\right).
\end{align*}

\subsection{Gaussian raking}

Since the raking iterations preserve Gaussianity, we define the parameter of the raking density at the $t$-th iteration as
\begin{align*}
\xi_t := (\mu_t, \Sigma_t),
\end{align*}
where $\mu_t$ and $\Sigma_t$ are the mean vector and the covariance matrix of $f^{(t)}$, respectively.
Then the raking recursion can be written as
\begin{align}
\xi_{t+1} = F(\xi_t,\eta) = H_2(H_1(\xi_t,\gamma_1),\gamma_2), \label{eq:gaussian-raking-iteration}
\end{align}
where $\eta = (\gamma_1',\gamma_2',(\gamma^s)')'$ is the vector of parameters of the input densities $f_1,f_2,f^s$, $H_1$ is the projection that enforces the first marginal, and $H_2$ is the projection that enforces the second marginal.

Let $\xi=(\mu,\Sigma)$.
It is straightforward to show that 
\begin{align*}
H_1(\xi,\gamma_1^\star)
=
\left(
\binom{\mu_1^\star}{\mu_2+A(\xi)\bigl(\mu_1^\star-\mu_1\bigr)},
\begin{pmatrix}
\Sigma_1^\star & \Sigma_1^\star A(\xi)'\\
A(\xi)\Sigma_1^\star & V(\xi)+A(\xi)\Sigma_1^\star A(\xi)'
\end{pmatrix}
\right),
\end{align*}
where $\gamma_1^\star=(\mu_1^\star,\Sigma_1^\star)$, and
\begin{align*}
    A(\xi):=\Sigma_{21}\Sigma_{11}^{-1},
\qquad
V(\xi):=\Sigma_{22}-\Sigma_{21}\Sigma_{11}^{-1}\Sigma_{12}.
\end{align*}
Similarly,
\[
H_2\bigl(\xi,\gamma_2^\star\bigr)
=
\left(
\begin{pmatrix}
\mu_1 + B(\xi)(\mu_2^\star-\mu_2)\\
\mu_2^\star
\end{pmatrix},
\begin{pmatrix}
U(\xi)+B(\xi)\Sigma_2^\star B(\xi)' &
B(\xi)\Sigma_2^\star\\
\Sigma_2^\star B(\xi)' &
\Sigma_2^\star
\end{pmatrix}
\right),
\]
where $\gamma_2^\star=(\mu_2^\star,\Sigma_2^\star)$ and
\[
B(\xi)
:=
\Sigma_{12}\Sigma_{22}^{-1},
\qquad
U(\xi)
:=
\Sigma_{11}
-
\Sigma_{12}\Sigma_{22}^{-1}\Sigma_{21}.
\]

\subsection{Derivatives of Gaussian raking parameters}

The Gaussian score is
\begin{align*}
    s_{t}(z)
:=\nabla_{\gamma_{t}}\log f^{(t)}(z)
=
\begin{pmatrix}
\Sigma_{t}^{-1}(z-\mu_{t})\\[2mm]
\frac12 D_d' \operatorname{vec}\!\Big(
\Sigma_{t}^{-1}\big[(z-\mu_{t})(z-\mu_{t})'-\Sigma_{t}\big]\Sigma_{t}^{-1}
\Big)
\end{pmatrix},
\end{align*}
where $D_d$ is the duplication matrix of order $d$, i.e., the unique matrix such that
\[
\operatorname{vec}(A) = D_d \operatorname{vech}(A)
\]
for any $d \times d$ symmetric matrix $A$.
Hence, the desired score is
\begin{align*}
\nabla_\eta \log f^{(t)}(z)
=
J_t(\eta)'
\begin{pmatrix}
\Sigma_t^{-1}(z-\mu_t)\\[2mm]
\frac12 D_d' \operatorname{vec}\!\Big(
\Sigma_t^{-1}\big[(z-\mu_t)(z-\mu_t)'-\Sigma_t\big]\Sigma_t^{-1}
\Big),
\end{pmatrix}
\end{align*}
where $J_t$ is generated by the recursion below, starting from $J_0=(0,0,I)$.

\begin{proposition}[Jacobian recursion for Gaussian raking]
Define
\[
J_t := \frac{\partial \xi_t}{\partial \eta'}
=
\begin{pmatrix}
\frac{\partial \xi_t}{\partial \gamma_1'} &
\frac{\partial \xi_t}{\partial \gamma_2'} &
\frac{\partial \xi_t}{\partial (\gamma^s)'}
\end{pmatrix}.
\]
Then \(J_t\) satisfies
\[
J_{t+1}
=
H_{2,\xi}\,H_{1,\xi_t}\,J_t
+
\begin{pmatrix}
H_{2,\xi}\,H_{1,\gamma_1} &
H_{2,\gamma_2} &
0
\end{pmatrix},
\]
with $J_0=
\begin{pmatrix}
0 & 0 & I
\end{pmatrix}$.
\end{proposition}

\begin{proof}
Applying the chain rule with respect to \(\eta\) to the raking iteration \eqref{eq:gaussian-raking-iteration} gives
\[
\frac{\partial \xi_{t+1}}{\partial \eta'}
=
\frac{\partial F(\xi_t,\eta)}{\partial \xi_t'}
\frac{\partial \xi_t}{\partial \eta'}
+
\frac{\partial F(\xi_t,\eta)}{\partial \eta'}.
\]
With \(J_t=\partial \xi_t/\partial \eta'\), this becomes
\[
J_{t+1}=F_{\xi_t}J_t+F_\eta.
\]

Since \(F=H_2\circ H_1\), another application of the chain rule yields
\[
F_{\xi_t}=H_{2,\xi}\,H_{1,\xi_t}.
\]
Moreover, \(\gamma_1\) enters only through the first projection, \(\gamma_2\) enters only through the second projection, and \(\gamma^s\) enters only through the initial condition. Therefore,
\[
F_\eta
=
\begin{pmatrix}
H_{2,\xi}\,H_{1,\gamma_1} &
H_{2,\gamma_2} &
0
\end{pmatrix}.
\]
Finally, we have
\[
\xi_0=\gamma^s,
\]
and hence
\[
J_0=
\frac{\partial \xi_0}{\partial \eta'}
=
\begin{pmatrix}
0 & 0 & I
\end{pmatrix}.
\]
\end{proof}

\begin{proposition}[Derivatives of the Gaussian projection maps]
Let
\[
\xi=(\mu,\Sigma), \qquad
\mu=
\begin{pmatrix}
\mu_1\\
\mu_2
\end{pmatrix},
\qquad
\Sigma=
\begin{pmatrix}
\Sigma_{11} & \Sigma_{12}\\
\Sigma_{21} & \Sigma_{22}
\end{pmatrix},
\]
and let the target marginal parameters be
\[
\gamma_1=(\bar\mu_1,\bar\Sigma_1),
\qquad
\gamma_2=(\bar\mu_2,\bar\Sigma_2).
\]
Define the first projection map
\[
H_1(\xi,\gamma_1)=:\xi^{(1)}=(\mu^{(1)},\Sigma^{(1)}),
\]
with
\[
A=\Sigma_{21}\Sigma_{11}^{-1},
\qquad
V=\Sigma_{22}-\Sigma_{21}\Sigma_{11}^{-1}\Sigma_{12},
\]
so that
\[
\mu^{(1)}=
\begin{pmatrix}
\bar\mu_1\\
\mu_2+A(\bar\mu_1-\mu_1)
\end{pmatrix},
\qquad
\Sigma^{(1)}=
\begin{pmatrix}
\bar\Sigma_1 & \bar\Sigma_1 A'\\
A\bar\Sigma_1 & V+A\bar\Sigma_1A'
\end{pmatrix}.
\]
Define the second projection map
\[
H_2(\xi^{(1)},\gamma_2)=:(\mu^{(2)},\Sigma^{(2)}),
\]
with
\[
B=\Sigma_{12}^{(1)}(\Sigma_{22}^{(1)})^{-1},
\qquad
U=\Sigma_{11}^{(1)}-\Sigma_{12}^{(1)}(\Sigma_{22}^{(1)})^{-1}\Sigma_{21}^{(1)},
\]
so that
\[
\mu^{(2)}=
\begin{pmatrix}
\mu_1^{(1)}+B(\bar\mu_2-\mu_2^{(1)})\\
\bar\mu_2
\end{pmatrix},
\qquad
\Sigma^{(2)}=
\begin{pmatrix}
U+B\bar\Sigma_2B' & B\bar\Sigma_2\\
\bar\Sigma_2B' & \bar\Sigma_2
\end{pmatrix}.
\]
Then the differentials of the auxiliary quantities are
\begin{align}
dA
&=
(d\Sigma_{21}-A\,d\Sigma_{11})\Sigma_{11}^{-1},
\label{eq:dA}
\\
dV
&=
d\Sigma_{22}
-d\Sigma_{21}\Sigma_{11}^{-1}\Sigma_{12}
-A\,d\Sigma_{12}
+A\,d\Sigma_{11}\Sigma_{11}^{-1}\Sigma_{12},
\label{eq:dV}
\\
dB
&=
(d\Sigma_{12}^{(1)}-B\,d\Sigma_{22}^{(1)})(\Sigma_{22}^{(1)})^{-1},
\label{eq:dB}
\\
dU
&=
d\Sigma_{11}^{(1)}
-dB\,\Sigma_{21}^{(1)}
-B\,d\Sigma_{21}^{(1)}.
\label{eq:dU}
\end{align}
Consequently, the derivatives of $H_1$ and $H_2$ are given by
\begin{align}
d\mu^{(1)}
&=
\begin{pmatrix}
d\bar\mu_1\\
d\mu_2+dA(\bar\mu_1-\mu_1)+A(d\bar\mu_1-d\mu_1)
\end{pmatrix},
\label{eq:dmu1}
\\
d\Sigma^{(1)}
&=
\begin{pmatrix}
d\bar\Sigma_1 &
d\bar\Sigma_1A' + \bar\Sigma_1 dA'\\
dA\,\bar\Sigma_1 + A\,d\bar\Sigma_1 &
dV + dA\,\bar\Sigma_1A' + A\,d\bar\Sigma_1A' + A\bar\Sigma_1 dA'
\end{pmatrix},
\label{eq:dSigma1}
\\
d\mu^{(2)}
&=
\begin{pmatrix}
d\mu_1^{(1)} + dB(\bar\mu_2-\mu_2^{(1)}) + B(d\bar\mu_2-d\mu_2^{(1)})\\
d\bar\mu_2
\end{pmatrix},
\label{eq:dmu2}
\\
d\Sigma^{(2)}
&=
\begin{pmatrix}
dU + dB\,\bar\Sigma_2B' + B\,d\bar\Sigma_2\,B' + B\bar\Sigma_2 dB' &
dB\,\bar\Sigma_2 + B\,d\bar\Sigma_2\\
d\bar\Sigma_2B' + \bar\Sigma_2 dB' &
d\bar\Sigma_2
\end{pmatrix}.
\label{eq:dSigma2}
\end{align}
The Jacobians of $H_1$ and $H_2$ are obtained by collecting the coefficients of the differentials in
\eqref{eq:dmu1}--\eqref{eq:dSigma2}, after converting $d\Sigma$ to $d\vech(\Sigma)$.
\end{proposition}

\begin{proof}
We prove the formulas for $H_1$; the derivation for $H_2$ is analogous.

First, note that for any invertible matrix $M$,
\[
d(M^{-1})=-M^{-1}(dM)M^{-1}.
\]
Since
\[
A=\Sigma_{21}\Sigma_{11}^{-1},
\]
we have
\[
dA
=
d\Sigma_{21}\Sigma_{11}^{-1}
+\Sigma_{21}d(\Sigma_{11}^{-1})
=
d\Sigma_{21}\Sigma_{11}^{-1}
-\Sigma_{21}\Sigma_{11}^{-1}(d\Sigma_{11})\Sigma_{11}^{-1},
\]
establishing \eqref{eq:dA}. Similarly,
\[
V=\Sigma_{22}-\Sigma_{21}\Sigma_{11}^{-1}\Sigma_{12},
\]
so
\begin{align*}
dV
&=
d\Sigma_{22}
-d\Sigma_{21}\Sigma_{11}^{-1}\Sigma_{12}
-\Sigma_{21}d(\Sigma_{11}^{-1})\Sigma_{12}
-\Sigma_{21}\Sigma_{11}^{-1}d\Sigma_{12} \\
&=
d\Sigma_{22}
-d\Sigma_{21}\Sigma_{11}^{-1}\Sigma_{12}
+A\,d\Sigma_{11}\Sigma_{11}^{-1}\Sigma_{12}
-A\,d\Sigma_{12},
\end{align*}
establishing \eqref{eq:dV}.

Differentiating the expression for $\mu^{(1)}$ yields \eqref{eq:dmu1}. Differentiating the block representation of $\Sigma^{(1)}$ and using the product rule together with \eqref{eq:dA}--\eqref{eq:dV} gives \eqref{eq:dSigma1}.
\end{proof}

\section{The discrete case}\label{sec:discrete}

In this section, we specialize our testing procedure to the case where the distributions $f_1,f_2,f^s$ are finitely discrete.

\subsection{Setup and notation}

Without loss of generality, we assume that \(Z_1\) and \(Z_2\) have finite supports \(\{1,\ldots,A\}\) and \(\{1,\ldots,B\}\), respectively.
Define
\[
f_1(a)=\Pb(Z_1=a),
\qquad
f_2(b)=\Pb(Z_2=b),
\qquad
f_s(a,b)=\Pb(Z_1=a,Z_2=b\mid S=1).
\]
Let
\[
f_{s,1}(a)=\sum_{b=1}^B f_s(a,b)
\]
and define
\[
f_s(b\mid a)
=
\frac{f_s(a,b)}{f_{s,1}(a)}
=
\Pb(Z_2=b\mid Z_1=a,S=1).
\]
Denote the corresponding \(A\times B\) matrix by
\[
P_s=\bigl(f_s(b\mid a)\bigr)_{a,b}
\]
and denote the $a$-th row of $P_s$ by
\[
c_a
=
\begin{pmatrix}
f_s(1\mid a)\\
\vdots\\
f_s(B\mid a)
\end{pmatrix}
.
\]
Finally, denote
\[
\pi_a=\Pb(S=1\mid Z_1=a).
\]

The MAR joint distribution is
\[
f_{\mathrm{MAR}}(a,b)
=
f_1(a)f_s(b\mid a),
\]
and its second marginal is
\[
t_b
=
\sum_{a=1}^A f_1(a)f_s(b\mid a),
\qquad
t=P_s'f_1.
\]
The hypothesis $H_0: f_{\mar} = f_{\an}$ can be formulated equivalently as
\[
H_0: t=f_2.
\]

Write
\[
F_s=\bigl(f_s(a,b)\bigr)_{a,b},
\qquad
D_1=\operatorname{diag}(f_1),
\qquad
D_2=\operatorname{diag}(f_2),
\]
and define
\[
\gamma
=
\begin{pmatrix}
f_1\\
f_2\\
\operatorname{vec}(F_s)
\end{pmatrix}.
\]
Notice that this is a redundant parameterization because
\[
\sum_{a=1}^A f_1(a)=\sum_{b=1}^B f_2(b) = \sum_{a=1}^A \sum_{b=1}^B f_s(a,b)=1.
\]
We consider the standard frequency estimators of the distributions $f_1,f_2,f_s$, 
\begin{align*}
    \widehat f_1(a) &= \frac1n\sum_{i=1}^n \mathbf 1(Z_{1i}=a),\\
    \widehat f_2(b) &= \frac1{n_r}\sum_{j=1}^{n_r} \mathbf 1(Z_{2j}^r=b), \\
    \widehat f_s(a,b) &= \frac{
        \sum_{i=1}^n
        S_i\mathbf 1(Z_{1i}=a, Z_{2i}=b)
        }{
        \sum_{i=1}^n
        S_i
        }.
\end{align*}
The estimator of the second marginal under MAR is
\[
\widehat t
=
\sum_{a=1}^A\widehat f_1(a)\widehat c_a
=
\widehat P_s'\widehat f_1,
\]
where \(e_b\) is the \(b\)-th standard basis vector in \(\mathbb R^B\) and
\[
\widehat c_a
=
\frac{
\sum_{i=1}^n
S_i\mathbf 1(Z_{1i}=a)e_{Z_{2i}}
}{
\sum_{i=1}^n
S_i\mathbf 1(Z_{1i}=a)
}.
\]

Given an estimator
\[
\widehat\gamma
=
\begin{pmatrix}
\widehat f_1\\
\widehat f_2\\
\operatorname{vec}(\widehat F_s)
\end{pmatrix},
\]
let \(\widehat f_{\mathrm{AN}}=f_{\mathrm{AN}}(\widehat\gamma)\) and
\[
\widehat f_{\mathrm{MAR}}(a,b)
=
\widehat f_1(a)\widehat f_s(b\mid a).
\]
The test statistic is
\[
\widehat T_n
=
\frac12
\sum_{a=1}^A\sum_{b=1}^B
\left(
\sqrt{\widehat f_{\mathrm{AN}}(a,b)}
-
\sqrt{\widehat f_{\mathrm{MAR}}(a,b)}
\right)^2.
\]

\subsection{The asymptotic distribution}

\begin{proposition}
\label{thm:discrete-hellinger-null}
Suppose that
\[
\frac{n}{n_r}\longrightarrow\kappa\in(0,\infty).
\]
Let
\[
\mathcal S = D_2-P_s'D_1P_s = \sum_{a=1}^A f_1(a) \left\{ \operatorname{diag}(c_a)-c_ac_a' \right\}
\]
and
\[
V=V_{\mathrm{pan}}+\kappa\Sigma_2,
\]
where
\[
\Sigma_2=D_2-f_2f_2'
\]
and
\[
V_{\mathrm{pan}} =
\sum_{a=1}^A
f_1(a)(c_a-f_2)(c_a-f_2)'
+
\sum_{a=1}^A
\frac{f_1(a)}{\pi_a}
\left\{
\operatorname{diag}(c_a)-c_ac_a'
\right\}.
\]
Then, under the null of MAR,
\[
n\widehat T_n
\rightsquigarrow
\frac18G'\mathcal S^+G,
\qquad
G\sim N(0,V),
\]
where \(\mathcal S^+\) is the Moore--Penrose inverse of
\(\mathcal S\).
\end{proposition}

\begin{proof}
The result follows from \Cref{prop:hellinger-reduction,prop:marginal-discrepancy-clt} and the continuous mapping theorem.
\end{proof}

\begin{lemma}\label{prop:hellinger-reduction}
Define
\[
\mathcal S=D_2-P_s'D_1P_s
\]
and let \(\mathcal S^+\) be the Moore--Penrose inverse of \(\mathcal S\).
Then
\[
n\widehat T_n
=
\frac18
\left\{
\sqrt n(\widehat t-\widehat f_2)
\right\}'
\mathcal S^+
\left\{
\sqrt n(\widehat t-\widehat f_2)
\right\}
+
o_p(1).
\]
\end{lemma}

\begin{proof}

Write
\[
\widehat T_n=T(\widehat\gamma),
\]
where
\[
T(\gamma)
=
\frac12
\sum_{a=1}^A\sum_{b=1}^B
\left(
\sqrt{f_{\mathrm{AN}}(a,b;\gamma)}
-
\sqrt{f_{\mathrm{MAR}}(a,b;\gamma)}
\right)^2.
\]
Under \(H_0\),
\[
T(\gamma_0)=0
\qquad\text{and}\qquad
DT(\gamma_0)=0.
\]
Let
\[
\delta_n=\widehat\gamma-\gamma_0.
\]
Because \(\delta_n=O_p(n^{-1/2})\), the second-order expansion of \(T\)
around \(\gamma_0\) gives
\[
\widehat T_n
=
\frac12D^2T(\gamma_0)[\delta_n,\delta_n]
+
o_p(n^{-1}).
\]

Let
\[
J=Dd(\gamma_0)
\]
denote the derivative of
\[
d(\gamma)=P_s'f_1-f_2.
\]
We have
\[
D^2T(\gamma_0)[h,h]
=
\frac14(Jh)'\mathcal S^+(Jh)
\]
for every direction \(h\) in the simplex tangent space. Hence,
\[
\widehat T_n
=
\frac18
(J\delta_n)'\mathcal S^+(J\delta_n)
+
o_p(n^{-1}).
\]

A first-order expansion of \(d\) around \(\gamma_0\) yields
\[
d(\widehat\gamma)
=
d(\gamma_0)+J\delta_n+r_n,
\qquad
r_n=o_p(n^{-1/2}).
\]
Under \(H_0\), \(d(\gamma_0)=0\), while
\[
d(\widehat\gamma)
=
\widehat P_s'\widehat f_1-\widehat f_2
=
\widehat t-\widehat f_2.
\]
It follows that
\[
J\delta_n
=
\widehat t-\widehat f_2-r_n.
\]

Moreover,
\[
\widehat t-\widehat f_2=O_p(n^{-1/2}).
\]
Since \(\mathcal S^+\) is a fixed matrix,
\[
(J\delta_n)'\mathcal S^+(J\delta_n)
-
(\widehat t-\widehat f_2)'
\mathcal S^+
(\widehat t-\widehat f_2)
=
-2(\widehat t-\widehat f_2)'\mathcal S^+r_n
+
r_n'\mathcal S^+r_n
=
o_p(n^{-1}).
\]
Substituting this result into the expansion of \(\widehat T_n\) completes the proof.
\end{proof}

\begin{lemma}
\label{prop:marginal-discrepancy-clt}

Suppose that
\[
\frac{n}{n_r}\longrightarrow\kappa\in(0,\infty)
\]
and
\[
f_1(a)>0
\qquad\text{and}\qquad
\pi_a>0
\]
for every \(a \in \{1,\dots,A\}\). Then, under $H_0$,
\[
\sqrt n
(\widehat t-\widehat f_2)
\rightsquigarrow
G,
\qquad
G\sim N(0,V),
\]
where
\[
V=V_{\mathrm{pan}}+\kappa\Sigma_2,
\]
\[
\Sigma_2
=
\operatorname{diag}(f_2)-f_2f_2',
\]
and
\[
V_{\mathrm{pan}}
=
\sum_{a=1}^A
f_1(a)(c_a-f_2)(c_a-f_2)'
+
\sum_{a=1}^A
\frac{f_1(a)}{\pi_a}
\left\{
\operatorname{diag}(c_a)-c_ac_a'
\right\}.
\]
\end{lemma}

\begin{proof}

We first derive an asymptotically linear representation of
\(\widehat t\). For each \(a\), the usual expansion of a conditional
sample proportion gives
\[
\sqrt n(\widehat c_a-c_a)
=
\frac1{\sqrt n}
\sum_{i=1}^n
\frac{
S_i\mathbf 1(Z_{1i}=a)
}{
f_1(a)\pi_a
}
\left(e_{Z_{2i}}-c_a\right)
+
o_p(1).
\]
Also,
\[
\sqrt n\{\widehat f_1(a)-f_1(a)\}
=
\frac1{\sqrt n}
\sum_{i=1}^n
\left\{
\mathbf 1(Z_{1i}=a)-f_1(a)
\right\}.
\]

Expanding
\[
\widehat t=\sum_{a=1}^A\widehat f_1(a)\widehat c_a
\]
around \((f_1,c_1,\ldots,c_A)\) yields
\[
\widehat t-t
=
\sum_{a=1}^A
c_a\{\widehat f_1(a)-f_1(a)\}
+
\sum_{a=1}^A
f_1(a)(\widehat c_a-c_a)
+
o_p(n^{-1/2}).
\]
Consequently,
\[
\sqrt n(\widehat t-t)
=
\frac1{\sqrt n}\sum_{i=1}^n\psi_i+o_p(1),
\]
where
\[
\psi_i
=
c_{Z_{1i}}-t
+
\frac{S_i}{\pi_{Z_{1i}}}
\left(e_{Z_{2i}}-c_{Z_{1i}}\right).
\]

The influence function has mean zero. Indeed, conditional on
\(Z_{1i}=a\),
\[
\E\left[
\frac{S_i}{\pi_a}
\left(e_{Z_{2i}}-c_a\right)
\,\middle|\,
Z_{1i}=a
\right]
=0.
\]
It follows that the two terms in \(\psi_i\) are uncorrelated. Moreover,
\[
Var(c_{Z_{1i}}-t)
=
\sum_{a=1}^A
f_1(a)(c_a-t)(c_a-t)'
\]
and
\[
\begin{aligned}
Var\left(
\frac{S_i}{\pi_{Z_{1i}}}
\left(e_{Z_{2i}}-c_{Z_{1i}}\right)
\right)
=
\sum_{a=1}^A
\frac{f_1(a)}{\pi_a}
\left\{
\operatorname{diag}(c_a)-c_ac_a'
\right\}.
\end{aligned}
\]
Therefore,
\[
Var(\psi_i)=V_{\mathrm{pan}},
\]
and the central limit theorem gives
\[
\sqrt n(\widehat t-t)
\rightsquigarrow
G_{\mathrm{pan}},
\qquad
G_{\mathrm{pan}}\sim N(0,V_{\mathrm{pan}}).
\]
On the other hand,
\[
\sqrt{n_r}(\widehat f_2-f_2)
=
\frac1{\sqrt{n_r}}
\sum_{j=1}^{n_r}
\left(e_{Z_{2j}^r}-f_2\right)
\rightsquigarrow
G_2,
\]
where
\[
G_2\sim N(0,\Sigma_2),
\qquad
\Sigma_2
=
\operatorname{diag}(f_2)-f_2f_2'.
\]
Since \(n/n_r\to\kappa\),
\[
\sqrt n(\widehat f_2-f_2)
\rightsquigarrow
\sqrt\kappa\,G_2.
\]
Since \(G_{\mathrm{pan}}\) and \(G_2\) are independent,
\[
\sqrt n
\left[
(\widehat t-\widehat f_2)-(t-f_2)
\right]
=
\sqrt n(\widehat t-t) - \sqrt n(\widehat f_2-f_2) \rightsquigarrow
G_{\mathrm{pan}}-\sqrt\kappa\,G_2 =:G.
\]
Hence,
\[
G\sim N(0,V),
\qquad
V=V_{\mathrm{pan}}+\kappa\Sigma_2.
\]
\end{proof}

\section{Construction of the cognition variable}\label{sec:app-cognition}


The UAS administers several cognitive tests every two years. To create a summary measure of cognition, we consider five cognitive assessments: Serial Sevens (sequentially subtracting 7 from 100 five times), Picture Vocabulary (presenting participants with a series of pictures and asking them to type the names of each one), Verbal Analogies (asking participants to identify relationships between word pairs), Number Series (presenting participants with a series of four numbers with one number missing, and asking them to identify the missing number), and Numeracy (questions asked include ones addressing
probability and computations related to monetary spending). Each of these tests is described in more detail in \citet{gatz2026online}.

Using cognition in longitudinal analysis requires one to correct for practice effects. Having seen a test before tends to improve performance on a test. A simple model of practice effects is
\[
y_{it}=\alpha_t+f(age_{it})+\mu_i +\epsilon_{it},
\]
where $y_{it}$ is a vector of cognitive scores (five in our case) for respondent $i$ at replication $t$. Hence, $t=1$ means the respondent sees the test for the first time; $t=2$ means it is the second time and so forth. The vector $\alpha_t$ represents the practice effect for replication $t$. We assume that scores change with age. We have approximated the function $f(age_{it})$ by a fourth-degree polynomial; $\mu_i$ is an individual effect (the extent to which an individual's cognition in a domain deviates from the average age pattern).

Estimation of the equation is straightforward. Denoting estimates by carets, $\widehat{x}_{it}:=y_{it}-\widehat{\alpha}_t$ can be seen as an estimate corrected for practice effects and age. To distill one overall cognition variable out of the vector $x_{it}$, we perform a confirmatory factor analysis, i.e., we estimate the model
\[
\widehat{x}_{it}=\gamma\xi_{it}+\nu_{it},
\]
where $\xi_{it}$ is a scalar (cognition of individual $i$) and the error covariance matrix of $\nu_{it}$  is assumed to be diagonal. The variance of $\xi_{it}$ is normalized to one. Adopting the convention of retaining factors with eigenvalues greater than one suggests that one factor can explain the five cognitive domains (the largest eigenvalue is 2.27, while the second largest eigenvalue equals 0.12).

\end{document}